\documentclass[10pt, final, journal, letterpaper, twocolumn]{IEEEtran}
\makeatletter
\def\ps@headings{%
	\def\@oddhead{\mbox{}\scriptsize\rightmark \hfil \thepage}%
	\def\@evenhead{\scriptsize\thepage \hfil \leftmark\mbox{}}%
	\def\@oddfoot{}%
	\def\@evenfoot{}}
\makeatother 

\IEEEoverridecommandlockouts
\usepackage{bbm}
\usepackage{amsfonts}
\usepackage[dvips]{graphicx}
\usepackage{times}
\usepackage{cite}
\usepackage{amsmath}
\usepackage{array}
\usepackage{amssymb}

\usepackage{stfloats}
\usepackage{slashbox}
\usepackage{graphicx}
\usepackage{footnote}
\usepackage{booktabs}
\usepackage{array}
\usepackage{algorithm}
\usepackage{subeqnarray}
\usepackage{cases}
\usepackage{threeparttable}
\usepackage{color}
\usepackage{hyperref}
\usepackage{epstopdf}
\usepackage{algpseudocode}
\usepackage{bm}
\usepackage{multirow}
\usepackage[labelformat=simple]{subcaption}
\usepackage{adjustbox}

\newtheorem{theorem}{Theorem}

\usepackage[absolute]{textpos}
\def\BibTeX{{\rm B\kern-.05em{\sc i\kern-.025em b}\kern-.08em
		T\kern-.1667em\lower.7ex\hbox{E}\kern-.125emX}}

\begin{document}
	
	\begin{textblock}{18}(2,0.5)
	\centering
	\noindent R. Li, X. Shao, S. Sun, M. Tao, and R. Zhang, ``Beam scanning for integrated sensing and communication in IRS-aided mmWave systems," in \textit{Proc. IEEE SPAWC}, Sep. 2023.
\end{textblock}
	
	\title{Beam Scanning for Integrated Sensing and Communication in IRS-aided mmWave Systems}

	\author{\authorblockN{Renwang Li$^{1,2}$, Xiaodan Shao$^3$, Shu Sun$^{1,2}$, Meixia Tao$^{1,2}$, Rui Zhang$^{3,4}$}\\
		\authorblockA{$^1$Department of Electronic Engineering, Shanghai Jiao Tong University, Shanghai, China		} \\		
		\authorblockA{$^2$Shanghai Key Laboratory of Digital Media Processing and Transmission, Shanghai Jiao Tong University, Shanghai, China		} \\
		\authorblockA{$^3$School of Science and Engineering, Shenzhen Research Institute of Big Data,
			The Chinese University of Hong Kong, Shenzhen, China} \\
		\authorblockA{$^4$Department of Electrical and Computer Engineering, National University of Singapore, Singapore} \\
		\authorblockA{Email: renwanglee@sjtu.edu.cn, shaoxiaodan@zju.edu.cn, shusun@sjtu.edu.cn, mxtao@sjtu.edu.cn, elezhang@nus.edu.sg}
		\thanks{This work is supported by the NSF of China under Grant 62125108 and Grant 62271310, by the Fundamental Research Funds for the Central Universities in China, and by the 2022 Stable Research Program of
			Higher Education of China under Grant 20220817144726001.}
	}

	\maketitle
	\thispagestyle{empty}
	
	\begin{abstract}
		This paper investigates an intelligent reflecting surface (IRS) aided millimeter-wave integrated sensing and communication (ISAC) system. Specifically, based on the passive beam scanning in the downlink, the IRS finds the optimal beam for reflecting the signals from the base station to a communication user. Meanwhile, the IRS estimates the angle of a nearby target based on its echo signal received by the sensing elements mounted on the IRS (i.e., semi-passive IRS). We propose an ISAC protocol  for achieving the above objective via simultaneous (beam) training and sensing (STAS). Then, we derive the achievable rate of the communication user and the Cramer-Rao bound (CRB) of the angle estimation for the sensing target in closed-form. {The achievable rate and CRB exhibit different performance against the duration of beam scanning. Specifically, the average achievable rate  initially rises and subsequently declines, while the CRB monotonically  decreases. Consequently, the duration of beam scanning should be carefully selected to balance communication and sensing performance.}  Simulation results have verified our analytical findings and shown that, thanks to the efficient use of downlink beam scanning signal for simultaneous communication and target sensing, the STAS protocol outperforms the benchmark protocol with orthogonal beam training and sensing.
	\end{abstract}
	
	\section{Introduction}
	Recently, integrated sensing and communication (ISAC) has been recognized as a key technology for the future sixth-generation (6G) wireless network due to its potential to enable sharing of spectrum and hardware resources between communication and sensing systems \cite{ 9737357}. Meanwhile,   millimeter-wave (mmWave) technology can provide high data rate for communication  as well as high resolution for sensing. Accordingly, mmWave is promising for realizing ISAC systems. However, mmWave is susceptible to   obstacles, and the performance of  mmWave ISAC systems will degrade dramatically in the absence of line-of-sight (LoS) path. To address this issue, intelligent reflecting surfaces (IRSs)  can be a practically viable solution \cite{8910627, 9326394}. An IRS is generally a digitally-controlled metasurface composed of a large number {of} passive reflecting elements (REs) that can independently impose phase shifts on the incident signal. By leveraging the IRS, a virtual LoS link can be created between two wireless nodes when their direct link is obstructed, allowing for uninterrupted  sensing and communication. 
	
	Motivated by the above appealing advantages, IRS-aided ISAC systems have been widely studied in various scenarios \cite{9364358, 9771801, 9416177, 10038557,9593143}.  The work \cite{9364358} studies the joint design of transmit beamforming at the base station (BS) and reflection coefficients at the IRS for maximizing  the signal-to-noise ratio (SNR) of radar detection and meeting the communication requirement at the same time. The works \cite{9771801} and \cite{9416177} consider a radar beampattern design problem  for single-user and multi-user  scenarios, respectively.  The authors in \cite{10038557} propose a feedback-based beam training approach to design the transmit beamforming and IRS reflection coefficients for conducting communication and sensing simultaneously in the data transmission period. The authors in \cite{9593143} propose a multi-stage hierarchical beam training codebook to achieve the desired localization accuracy of the target while ensuring a reliable communication link with the user. Notice that, all of the aforementioned works adopt passive IRS to assist sensing, and thus their performance  is hindered by the severe path loss of the BS-IRS-target-IRS-BS cascaded echo link, particularly in mmWave frequencies. Furthermore, most existing works focus on ISAC during the data transmission phase and do not consider the beam training phase, which is practically required for mmWave communication systems.
	\begin{figure}[t]
		\begin{centering}
			\includegraphics[width=.4\textwidth]{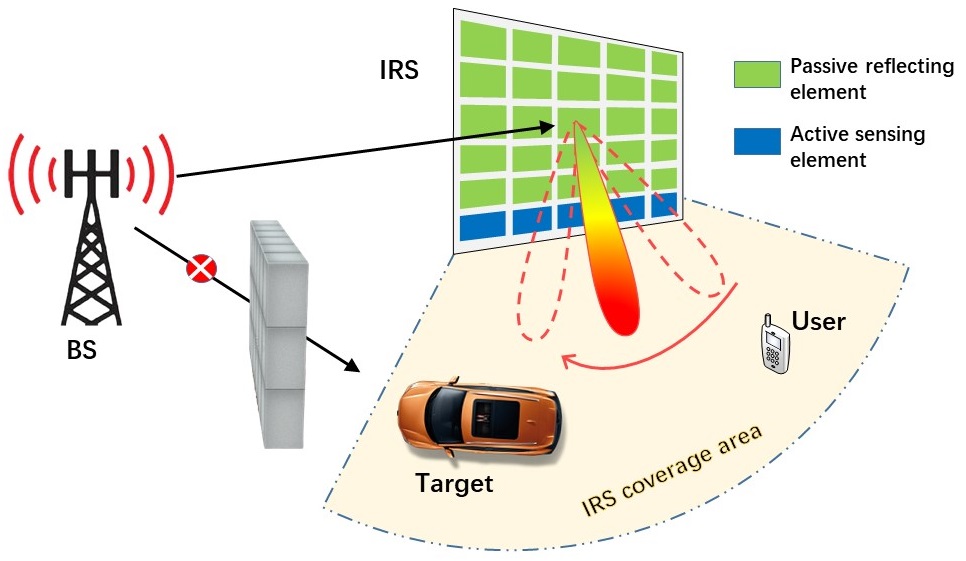}
			\caption{Semi-passive IRS aided ISAC.}\label{fig_system}
		\end{centering}
		\vspace{-0.8cm}
	\end{figure}	
	
	As such, in this paper, we consider an IRS-aided mmWave ISAC system, as illustrated in Fig. \ref{fig_system}, where a semi-passive  IRS consisting of passive REs and active sensing elements (SEs) is adopted to create  virtual LoS links for a communication user and a sensing target to maximize the data transmission rate and  the angle estimation accuracy, respectively. Note that the active SEs are used to collect the echoes reflected from the target for its angle estimation. Compared with a fully passive IRS, the semi-passive IRS can significantly reduce the path loss of the received echo signal by avoiding the IRS-BS link \cite{9724202}.
	To cater to practical mmWave systems, we propose a simultaneous
	(beam) training and sensing (STAS) protocol by exploiting the downlink IRS beam training/scanning for the communication user to estimate the target's angle concurrently.
	In addition, we analyze the achievable rate of the communication user and derive the Cramer-Rao bound (CRB) of the target angle estimation. Simulation results show that the achievable rate and CRB both depend on the time of beam scanning, but behaving differently.  {Specifically, the average achievable rate initially rises and subsequently declines, while the CRB consistently decreases.} Furthermore, we demonstrate that the STAS protocol outperforms the benchmark protocol with orthogonal beam training and sensing, due to its effective utilization of the downlink scanning signal for simultaneous sensing. 
	
	\section{System Model and Proposed Protocol}
	\subsection{System Model}
	We consider an ISAC system with the aid of a semi-passive IRS as illustrated in Fig. \ref{fig_system}, where an $N$-antenna BS attempts to communicate with a single-antenna user, and also to detect the angle of a sensing target from the IRS. The direct links between the BS and the user, as well as the target, are assumed to be blocked due to unfavorable propagation environment. Thus, the IRS is deployed to create virtual links for both communication and sensing. We consider the use of a semi-passive IRS consisting of $M$ passive REs to reflect the transmitted signals from the BS to the user and target, and $M_s$ active SEs to collect the echoes from the target for its angle estimation. The complex-valued baseband transmitted signal at the BS can be expressed as $\mathbf{x}=\mathbf{w}s$, where  $s$ denotes the training/data symbol for the communication user with unit power and $\mathbf{w} \in \mathbb{C}^{N\times 1}$ is the transmit beamforming vector with $\|\mathbf{w}\|^2=1$. Then, the received signal $y_u$ at the user can be expressed as
	\begin{equation} \label{exp_user}
		y_u = \sqrt{P_t}\mathbf{h}_u^H \operatorname{diag} (\boldsymbol{\phi}) \mathbf{G} \mathbf{w} s + n_u, \\
	\end{equation}
	where $P_t$ is the transmit power at the BS, $\mathbf{G}\in \mathbb{C}^{M\times N}$ represents the channel between the BS and IRS REs, $\mathbf{h}_u\in \mathbb{C}^{M\times 1}$ represents the channel between the IRS REs and the user, $n_u \sim \mathcal{CN}(0, \sigma^2)$ is the receiver AWGN with $\sigma^2$ representing the noise power, and $\boldsymbol{\phi}\in \mathbb{C}^{M\times 1}$ represents the reflection vector at the IRS REs, which can be written as
	\begin{equation} 
		\boldsymbol{\phi}=\left[e^{j\phi_1},e^{j\phi_2},\ldots, e^{j\phi_M}  \right],
	\end{equation}
	with $\phi_i$ being the phase shift by the $i$-th RE.
	
	 {We adopt the LoS channel model to characterize the mmWave channel. For ease of exposition, we assume that uniform linear arrays (ULAs)  are equipped at the BS, IRS REs, and IRS SEs.} Thus, the BS-IRS REs channel can be expressed as
	\begin{equation} 
		\mathbf{G} = \alpha_g \mathbf{a}_r (\theta_{BI}) \mathbf{a}_b^H (\vartheta_{BI}),
	\end{equation}
	where $\alpha_g=\frac{\lambda}{4\pi d_{BI}} e^{\frac{j 2\pi d_{BI}}{\lambda}}$ denotes the complex-valued path gain of the BS-IRS REs channel with $\lambda$ being the carrier wavelength and $d_{BI}$ being the distance between the BS and IRS,  $\vartheta_{BI}$ denotes the angle of departure (AoD) from the BS, $\theta_{BI}$ denotes the angle of arrival (AoA) to the IRS, and $\mathbf{a}_r(\cdot)$ $\left(\mathbf{a}_b(\cdot)\right)$ denotes the array response vector associated with the IRS (BS). The array response vector for a ULA with $M$ elements of half-wavelength spacing, and the center of the ULA  as the reference point, can be expressed as:
	\begin{equation} \label{exp_array_resp}
		{\mathbf{a}}(\theta) = \left[ 
		e^{-j \frac{(M-1)\pi \sin (\theta)}{2}} \enspace e^{-j \frac{(M-3)\pi \sin (\theta)}{2}}  \enspace \ldots \enspace e^{j \frac{(M-1)\pi \sin (\theta)}{2}}
		\right]^{T}.
	\end{equation}
	The IRS-user channel $\mathbf{h}_u$ can also be written as
	\begin{equation} 
		\mathbf{h}_u = \alpha_h \mathbf{a}_r (\theta_{IU} ),
	\end{equation}
	where $\alpha_h=\frac{\lambda}{4\pi d_{IU}} e^{\frac{j 2\pi d_{IU}}{\lambda}}$ denotes the complex-valued path gain of the IRS-user channel with $d_{IU}$ being the distance between the IRS and user, and  $\theta_{IU}$ denotes the AoD associated with the IRS.	
	 {Considering that the locations of the BS and IRS are fixed, the BS-IRS channel $\mathbf{G}$ is assumed to be constant for a long period, which can be estimated beforehand at the BS to achieve the optimal transmit beamforming as $\mathbf{w}=\frac{1}{\sqrt{N}} \mathbf{a}_b (\vartheta_{BI})$. Thus, in this paper we focus on the  beam training at the IRS.} And the received signal at the communication user  can be rewritten as
	\begin{equation} \label{exp_user2}
		y_u = \sqrt{N P_t} \alpha_g \mathbf{h}_u^H \operatorname{diag} (\boldsymbol{\phi}) \mathbf{a}_r (\theta_{BI}) s + n_u. \\
	\end{equation}
	
	The IRS SEs can simultaneously receive the signals transmitted  from  the BS and the echoes reflected by the target\footnote{The radar cross section (RCS) of the user is usually significantly smaller compared to the target. Hence, the echo signal reflected by the user can be safely ignored in the target angle estimation.}. In general, the angles of the target and BS with respect to the IRS are different and can be estimated by the IRS SEs based on the received echoes. The angle between the BS and IRS can be determined in advance by the IRS SEs, which facilitates in estimating the angle of the target. For simplicity, we only consider the echo signal reflected by the target. Hence, the received signal $\mathbf{y}_s \in \mathbb{C}^{M_s\times 1}$ at the IRS SEs is given by
		\begin{align} \label{exp_sensors}
			\mathbf{y}_s &= \sqrt{P_t}\mathbf{H}_t \operatorname{diag} (\boldsymbol{\phi}) \mathbf{G} \mathbf{w} s   + \mathbf{n}_s, \notag \\
			&= \sqrt{N P_t } \alpha_g \mathbf{H}_t \operatorname{diag} (\boldsymbol{\phi}) \mathbf{a}_r (\theta_{BI}) s  + \mathbf{n}_s,
		\end{align}
	where $\mathbf{H}_t \in \mathbb{C}^{M_s\times M}$ denotes the channel matrix on the IRS REs-target-IRS SEs link, and $\mathbf{n}_s \sim \mathcal{CN}(0, \sigma^2 \mathbf{I}_{M_s})$ is the receiver AWGN.
	The channel $\mathbf{H}_t$ can be represented as
	\begin{equation} 
		\mathbf{H}_t = \alpha_s \mathbf{a}_s ( \theta_{IT}) \mathbf{a}_r^H (\theta_{IT} ),
	\end{equation}
	where  $\theta_{IT}$ denotes the AoD from the IRS SEs, $\mathbf{a}_s(\cdot)$  represents the array response vector associated with the IRS  SEs,  $\alpha_s=\sqrt{\frac{\lambda^2 \kappa}{64\pi^3 d_{IT}^4}} e^{\frac{j 4\pi d_{IT}}{\lambda}}$ refers to the complex path gain of the IRS REs-target-IRS SEs link \cite{9367457},  in which $d_{IT}$ denotes the  distance between the IRS and target and $\kappa$ is the RCS of the target.
	
	\subsection{Proposed Protocol for ISAC}
	\begin{figure}[t]
		\begin{centering}
			\includegraphics[width=.47\textwidth]{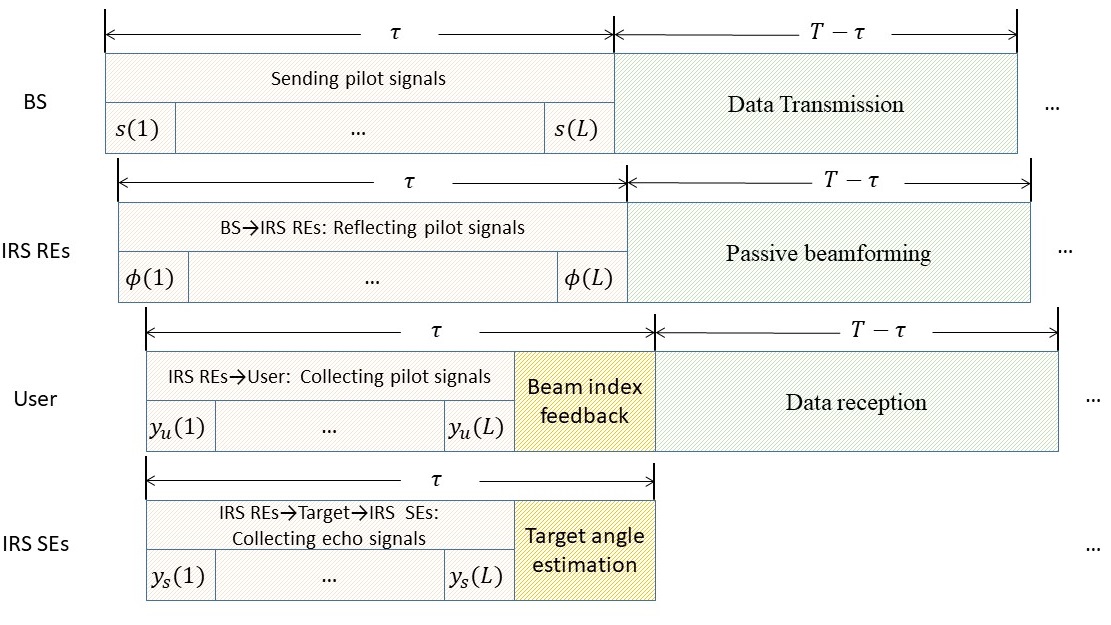}
			\caption{IRS-aided mmWave STAS protocol.}\label{fig_protocol}
		\end{centering}
		\vspace{0cm}
	\end{figure}
	In this subsection, we propose a practical protocol for the considered IRS-aided mmWave ISAC system. We adopt the widely-used discrete Fourier transformation (DFT) codebook $\mathbf{D}$ with $L$ beams as  follows
	\begin{equation} \label{exp_codebook}
		\mathbf{D} \triangleq[\mathbf{a}_r (\eta(1)), \mathbf{a}_r(\eta(2)), \cdots, \mathbf{a}_r(\eta(L))] \in \mathbb{C}^{M \times L},
	\end{equation}
	where $\eta(i)=\arcsin\left(-1+\frac{2 i-1}{L}\right), i=1, \cdots, L$, and $L \geq M$. Following the existing mmWave communication protocol, beam training/scanning needs to be first conducted at the IRS, followed by data transmission. Nevertheless, in our system model, the IRS SEs can exploit downlink beam scanning for sensing by collecting the echoes reflected from the target. We thus propose the STAS protocol, as shown in Fig. \ref{fig_protocol}, for enabling both downlink beam training and sensing simultaneously.  {The STAS protocol  is divided into two phases within the channel coherence time $T$: beam scanning phase with time duration of $\tau$, and data transmission phase with time duration of $T-\tau$.}
	\begin{itemize}
		\item Phase I (beam scanning): The BS sends downlink training signals. The IRS REs sweep the beams in a given codebook $\mathbf{D}$, while the IRS SEs collect the echo signal reflected from the target.
		At the end of IRS beam scanning, the communication user finds the IRS's best beam and feeds its index back to the IRS controller (directly or via the BS). In the meanwhile, the IRS SEs estimate the target angle based on the received echo signal.
		\item Phase II (data transmission): The BS sends downlink data signals to the communication user, and the IRS REs adopt  the best beam based on its feedback in Phase I, where the criterion for the best beam will be detailed in the next section.
	\end{itemize}		
	
	\section{Performance Analysis}

	\subsection{Achievable Communication Rate}
	 {In this subsection, we first analyze the best channel gain during the beam training phase, and then derive the achievable rate of the communication user.} In the period of beam training, the BS's signal can be set as $s =1$. The received signal at the user can be expressed as
		\begin{align}
			y_u(t) &= \sqrt{N P_t} \alpha_g \alpha_h \mathbf{a}_r^H (\theta_{IU}) \operatorname{diag} (\boldsymbol{\phi}(t)) \mathbf{a}_r (\theta_{BI})  + n_u(t)  \notag\\
			& = \sqrt{N P_t} \alpha_g \alpha_h \boldsymbol{\phi}^T(t) \operatorname{diag} (\mathbf{a}_r^H (\theta_{IU}) ) \mathbf{a}_r (\theta_{BI})  +n_u(t) \notag \\
			& = \sqrt{N P_t} \alpha_g \alpha_h \left(\mathbf{a}_r^H (\overline{\theta}_{IU}) \boldsymbol{\phi}(t) \right) +n_u(t),
		\end{align}
	where $\overline{\theta}_{IU}=\arcsin \left(\sin (\theta_{IU}) - \sin (\theta_{BI})\right)$, and $\boldsymbol{\phi}(t) \in \mathbf{D}$.  We assume that $\ell$ is the best beam index, then
	\begin{equation} 
	\ell = \arg \max \limits_{t, t=1, \cdots, L} |y_u(t)|^2.
	\end{equation}	
	Let $\delta= \left|\sin (\overline{\theta}_{IU}) - \sin (\eta(\ell)) \right| $ denote  the spatial direction difference between $\sin (\overline{\theta}_{IU})$ and its adjacent beam $\sin(\eta(\ell))$ $\left(0\leq \delta\leq \frac{1}{L}\right)$. Then, by denoting the best beam as $\boldsymbol{\phi}^\star=\mathbf{a}_r (\eta(\ell))$, the IRS beamforming gain can be expressed as
	\begin{equation} 
		\hspace{-0.2cm} \left|\mathbf{a}_r^H (\overline{\theta}_{IU}) \boldsymbol{\phi}^\star \right| = \left| \sum\limits_{m=1}^M e^{j\pi \delta \left(-\frac{M-1}{2}+m-1\right)} \right|  = \frac{\sin (\frac{\pi M \delta }{2})}{\sin (\frac{\pi  \delta }{2})}.
	\end{equation}
	The function $\frac{\sin (\frac{\pi M x }{2})}{\sin (\frac{\pi  x }{2})}$ exhibits behavior similar to that of the $sinc$ function and has a zero value at $\frac{2}{M}$. This function decreases monotonically over $[0,\frac{2}{M}]$.   Thus, when the user's angle $\overline{\theta}_{IU}$ is exactly aligned with the angle of the beams, the IRS beamforming gain reaches its maximum, i.e., $\delta=0$ and $\left|\mathbf{a}_r^H (\overline{\theta}_{IU}) \boldsymbol{\phi}^\star \right|=M$. When the user's angle $\overline{\theta}_{IU}$ lies in the middle of two adjacent beams, the IRS beamforming gain is the lowest, i.e., $\delta=\frac{1}{L}$ and  $\left|\mathbf{a}_r^T (\overline{\theta}_{IU}) \boldsymbol{\phi}^\star \right|={\sin \left(\frac{\pi M  }{2 L}\right)} {\sin^{-1} (\frac{\pi  }{2 L})}$.
	
	After beam training, the user finds the IRS's best beam $\boldsymbol{\phi }^\star$ and feeds its index back to the IRS controller. The IRS REs then adopt  the beam to reflect the signals from the BS to the communication user  during the data transmission phase (i.e., Phase II). The achievable rate of the user is thus given by
	\begin{equation} \label{exp_sweeping}
		R  = \frac{T-\tau}{T} \log \left(1+ \frac{N P_t |\alpha_g|^2 |\alpha_h|^2}{\sigma^2} \frac{\sin^2 (\frac{\pi M \delta  }{2  })}{\sin^2 (\frac{\pi \delta  }{2  })} \right),
	\end{equation}
	 Assuming that the duration of one beam is equal to $K$-symbol duration, we have $\tau=KL$.

	\subsection{CRB for Angle Estimation}
	In this subsection, the target angle is first estimated via the maximum likelihood estimation (MLE). Then, the CRB of the angle estimation is derived. 
	The received echo signals at the IRS SEs in \eqref{exp_sensors} can be represented as
		\begin{align} \label{exp_sensor_phaseI}
		\hspace{-0.2cm}	\mathbf{y}_s(t) &= \sqrt{P_t} \mathbf{H}_t \operatorname{diag} (\boldsymbol{\phi}(t)) \mathbf{G} \mathbf{w} s(t)   + \mathbf{n}_s(t) \notag \\
			& = \sqrt{N P_t} \alpha_g \alpha_s \mathbf{a}_s(\theta_{IT}) \mathbf{a}_r^H (\theta_{IT}) \operatorname{diag} (\boldsymbol{\phi}(t)) \mathbf{a}_r(\theta_{BI})  + \mathbf{n}_s(t) \notag \\
			& = \sqrt{N P_t} \alpha_g \alpha_s \mathbf{a}_s(\theta_{IT}) \boldsymbol{\phi}^T(t) \operatorname{diag} (\mathbf{a}_r^H (\theta_{IT}) ) \mathbf{a}_r (\theta_{BI}) + \mathbf{n}_s(t) \notag \\
			& = \sqrt{N P_t} \alpha_g \alpha_s \mathbf{a}_s(\theta_{IT})  \left(\mathbf{a}_r^T (\overline{\theta}_{IT}) \boldsymbol{\phi}(t) \right) +\mathbf{n}_s(t), 
		\end{align}
	where $\overline{\theta}_{IT}=\arcsin \left(\sin (\theta_{BI}) - \sin (\theta_{IT})\right)$.
	By collecting all received signals in the period of beam scanning, we have
		\begin{align} \label{exp_collect_sweeping}
			\mathbf{Y} &= [\mathbf{y}_s(1), \mathbf{y}_s(2),\cdots, \mathbf{y}_s(L) ] \notag \\
			& = \sqrt{N P_t} \alpha_g \alpha_s \mathbf{a}_s(\theta_{IT})   \mathbf{a}_r^T (\overline{\theta}_{IT}) \left[ \boldsymbol{\phi}(1), \boldsymbol{\phi}(2), \cdots, \boldsymbol{\phi}(L) \right]   +\mathbf{N} \notag \\
			& \triangleq \alpha_s \mathbf{a}_s (\theta_{IT}) \mathbf{q}^T(\theta_{IT}) \mathbf{X} + \mathbf{N},
		\end{align}
	where  $\mathbf{X} \triangleq \sqrt{N P_t} \alpha_g \left[ \boldsymbol{\phi}(1), \boldsymbol{\phi}(2), \cdots, \boldsymbol{\phi}(L) \right]$, $\mathbf{q}(\theta_{IT}) \triangleq \mathbf{a}_r (\overline{\theta}_{IT})$, and $\mathbf{N} \triangleq [\mathbf{n}_s(1), \mathbf{n}_s(2), \cdots, \mathbf{n}_s(L ) ]$. Let $\mathbf{R}_x  \triangleq \frac{1}{L} \mathbf{X} \mathbf{X}^H$ represent the covariance matrix of $\mathbf{X}$. As the codebook is designed as \eqref{exp_codebook}, we have $\mathbf{R}_x = N P_t |\alpha_g|^2 \mathbf{I}_M$. 
	For ease of notation, we simply re-denote $\theta_{IT}$ by $\theta$.   Let $\boldsymbol{\xi} =[\theta, \text{Re}\{\alpha_s\}, \text{Im}\{\alpha_s\}]\in \mathbb{R}^{3\times 1}$ denote the vector of the unknown parameters to be estimated, which includes the target's angle and the complex channel coefficients. By vectorizing \eqref{exp_collect_sweeping}, we have
	\begin{equation} 
		\operatorname{vec}(\mathbf{Y}) = \alpha_s \operatorname{vec}(\mathbf{u}(\theta))+ \operatorname{vec}(\mathbf{N}),
	\end{equation}
	where $\mathbf{u}(\theta) =   \mathbf{a}_s(\theta) \mathbf{q}^T(\theta)  \mathbf{X} $. Then, the target angle can be estimated according to the following theorem. 
	\begin{theorem}
		The angle estimated via the MLE is given by
		\begin{equation} \label{est_mle}
			\vspace{-0.2cm}
			\theta_\text{MLE} = \arg \max \limits_{\theta} \quad \left|\mathbf{a}_s^H (\theta) \mathbf{Y} \mathbf{X}^H \mathbf{q}^*(\theta) \right|^2,
		\end{equation}	
		which can be solved by exhaustive search over $\left[-\frac{\pi}{2}, \frac{\pi}{2} \right]$.
	\end{theorem}
	\begin{proof}
		See Appendix \ref{append_mle}.
	\end{proof} 
	
	In the following, we evaluate the performance of the angle estimation by deriving its CRB. Let $\mathbf{F}\in \mathbb{R}^{3\time 3}$ denote the Fisher information matrix (FIM) for estimating $\boldsymbol{\xi}$. Since $\mathbf{N} \in \mathcal{CN}(0, \mathbf{R}_z)$ with $\mathbf{R}_z= \sigma^2 \mathbf{I}_{M_s}$, each entry of $\mathbf{F}$ is given by \cite{kay1993fundamentals}   
	\begin{equation} 
		\vspace{-0.2cm}
		\mathbf{F}_{i,j} = 2 \operatorname{Re}\left\{\frac{\partial (\alpha_s \mathbf{u}(\theta))^H}{\partial \xi_i} \mathbf{R}_z^{-1} \frac{\partial (\alpha_s \mathbf{u}(\theta))}{\partial \xi_j}\right\}, i,j \in\{1,2,3\}	
	\end{equation}
	Then, the FIM $\mathbf{F}$ can be partitioned as
	\begin{equation}
		\vspace{-0.2cm}
		\tilde{\mathbf{F}}=\left[\begin{array}{cc}
			{\mathbf{F}}_{\theta \theta} & {\mathbf{F}}_{\theta \overline{\boldsymbol{\alpha}}} \\
			{\mathbf{F}}_{\theta \overline{\boldsymbol{\alpha}}}^{\mathrm{T}} & \tilde{\mathbf{F}}_{\overline{\boldsymbol{\alpha}} \overline{\boldsymbol{\alpha}}}
		\end{array}\right],
	\end{equation}
	where $\overline{\boldsymbol{\alpha}}= [\operatorname{Re}\{\alpha_s\},\operatorname{Im}\{\alpha_s\}]^T$. The CRB for estimating the angle $\theta$ is defined as
	\begin{equation}
		\vspace{-0.2cm}
		\text{CRB}(\theta) = [\mathbf{F}^{-1}]_{1,1}=[\mathbf{F}_{\theta\theta} - \mathbf{F}_{\theta \overline{\boldsymbol{\alpha}}} \mathbf{F}_{\overline{\boldsymbol{\alpha}} \overline{\boldsymbol{\alpha}}}^{-1} \mathbf{F}_{\theta \overline{\boldsymbol{\alpha}}}^T]^{-1}.
	\end{equation}
	Then, the CRB for target sensing is given by the following theorem.
	\begin{theorem}
		The CRB for estimating the angle $\theta$ is given by
		\begin{equation} \label{exp_crb_ori}
			\hspace{-0.15cm}	\mathrm{CRB}(\theta)
				=\frac{\sigma^2}{2  |\alpha_s|^2\left(\operatorname{tr}\left(\dot{\mathbf{u}}(\theta)  \dot{\mathbf{u}}^{\mathrm{H}}(\theta)\right) - \frac{\left|\operatorname{tr}\left(\mathbf{u}(\theta) \dot{\mathbf{u}}^{\mathrm{H}}(\theta)\right)\right|^2} {\operatorname{tr}\left(\mathbf{u}(\theta) \mathbf{u}^{\mathrm{H}}(\theta)\right)}\right)}.		
		\end{equation}
	\end{theorem}
	\begin{proof}
		See Appendix \ref{append_crb}.
	\end{proof}
	
	With the array response vector  defined as in  \eqref{exp_array_resp}, we obtain
	\begin{align} 
		&\mathbf{a}_s^H (\theta) \dot{\mathbf{a}}_s(\theta) = 0, \dot{\mathbf{a}}_s^H(\theta) \mathbf{a}_s(\theta) = 0, \\
		&\mathbf{q}^H (\theta) \dot{\mathbf{q}}(\theta) = 0, \dot{\mathbf{q}}^H(\theta) \mathbf{q}(\theta) = 0, \forall \theta.
	\end{align}
	Consequently, the CRB for estimating the angle $\theta$ with the STAS protocol can be simplified as
	\begin{equation} \label{exp_crb_phaseI}
		\mathrm{CRB}(\theta)=\frac{\sigma^2}{2 L N P_t |\alpha_s|^2 |\alpha_g|^2 \left( M \|\dot{\mathbf{a}}_s(\theta)\|^2 + M_s \|\dot{\mathbf{q}}(\theta)\|^2  \right) },
	\end{equation}
	where $\dot{\mathbf{a}}_s(\theta) \triangleq \frac{\partial  \mathbf{a}_s(\theta) }{\partial \theta} = j \pi \cos(\theta) \boldsymbol{\zeta} \mathbf{a}_s(\theta)$, $\boldsymbol{\zeta}=\operatorname{diag}\left(-\frac{M-1}{2}, -\frac{M-3}{2}, \cdots, \frac{M-1}{2}\right)$, and $\dot{\mathbf{q}}(\theta) \triangleq \frac{\partial  \mathbf{q}(\theta) }{\partial \theta} = j \pi \cos(\theta) \boldsymbol{\zeta} \mathbf{q}(\theta)$. For a given set of SNR, the number of BS antennas, BS-IRS path loss, and IRS-target-IRS path loss, we can improve the sensing accuracy by increasing the codebook size, the number of IRS REs, and the number of IRS SEs, according to \eqref{exp_crb_phaseI}. 
	
	\section{Simulation Results}

	\begin{figure*}[t] 
		\begin{minipage}[t]{0.33\linewidth} 
			\centering
			\includegraphics[width=2.4in, height=1.9in]{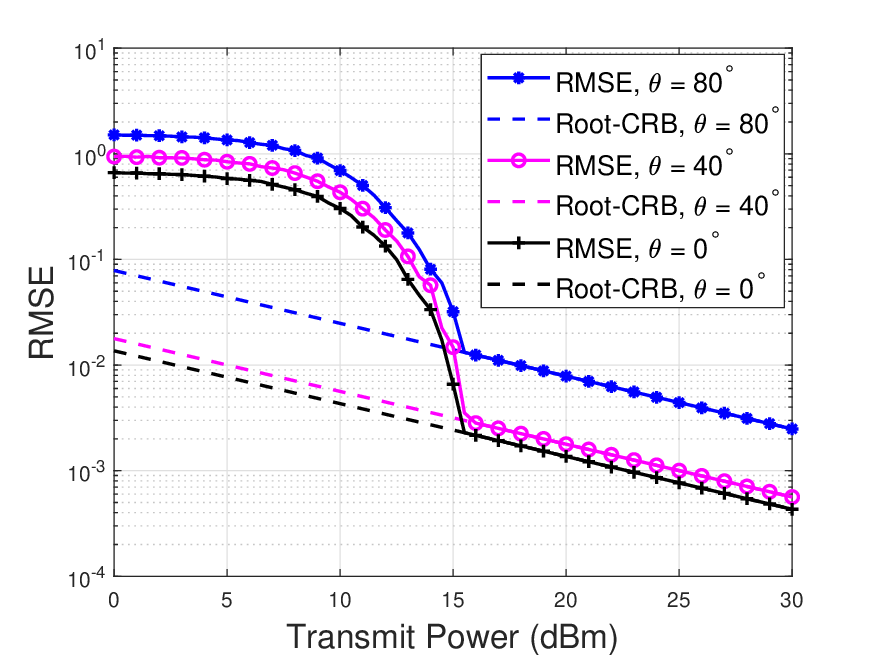} 
			\caption{RMSE and RCRB versus transmit power.} 
			\label{fig_crb} 
		\end{minipage}%
		\hspace{0.2cm}
		\begin{minipage}[t]{0.33\linewidth}
			\centering
			\includegraphics[width=2.4in, height=1.9in]{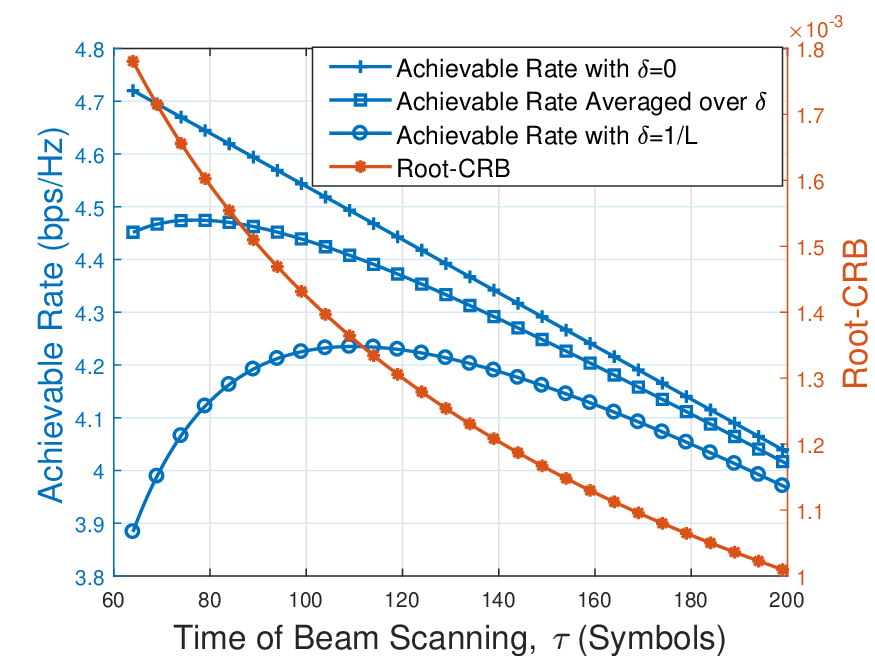}
			\caption{ Achievable rate and RCRB versus the time of IRS beam scanning.} \label{fig_online_m}
		\end{minipage}%
		\hspace{0.1cm}
		\begin{minipage}[t]{0.33\linewidth}
			\centering
			\includegraphics[width=2.4in, height=1.9in]{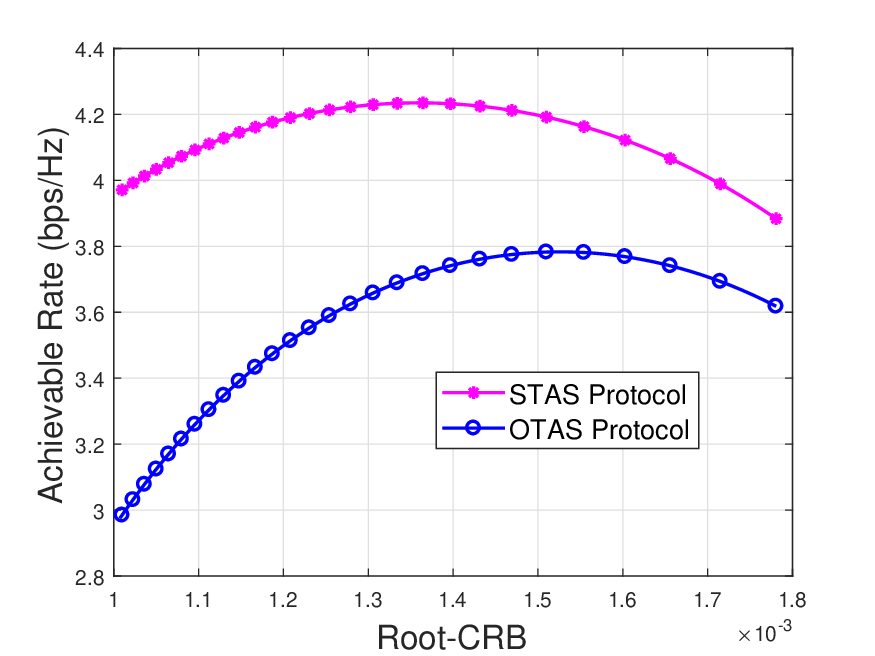}
\caption{ Achievable rate of the STAS and OTAS protocols versus RCRB.}\label{fig_online_w}
		\end{minipage}		
		\vspace{-0.3cm}
	\end{figure*}
	In this section, numerical examples are provided to validate the performance of the beam scanning-based ISAC system. The carrier frequency is $f_c=28$ GHz. Other system parameters are set as follows unless specified otherwise later: $N=64$, $M=64$, $M_s=8$, $L=64$, $K=1$, $P_t=20$ dBm, $T=1000$ symbols, $d_{BI}=30$ m, $d_{IU}=10$ m, $d_{IT}=5$ m, $\tau=64$ symbols, $\theta_{BI} = -30^\circ$, $\theta_{IT} = 40^\circ$, $\theta_{IU} = 0^\circ$, $\sigma^2 = -120$ dBm, and $\kappa = 7$ dBsm. The curves of root mean squared error (RMSE) are obtained by averaging over 1000 independent realizations of the noise.
	
	Fig. \ref{fig_crb} illustrates  the tightness of the derived root CRB (RCRB) of angle estimation by comparing it with RMSE. The RMSE is defined as $\text{RMSE} \triangleq \sqrt{\mathbb{E}[(\theta-\widetilde{\theta})^2]}$, where $\widetilde{\theta}$ denotes the estimation of $\theta$. It is observed that the derived RCRB is identical to the RMSE when the BS transmit power is larger than 15 dBm, which validates the accuracy of the  RCRB derivation. However, the RCRB is loose in the low SNR regime, owing to its inherent limitation of being determined by the local curvature of the log-likelihood function but not considering the global information \cite{9705498}. In addition, as $\theta$ increases, the norm $\|\dot{\mathbf{a}}_s(\theta)\|^2$ and $\|\dot{\mathbf{q}}(\theta)\|^2$ in  \eqref{exp_crb_phaseI} decreases, leading to the increase of RCRB.
	
	Next, we evaluate the performance of our proposed beam scanning-based ISAC system.  {Fig. \ref{fig_online_m} shows the achievable rate  and RCRB	versus the time of beam scanning. The curve of ``Achievable Rate Averaged over $\delta$" is generated by taking the expectation of \eqref{exp_sweeping} over $\delta$, where $\delta \sim \mathcal{U}(0,\frac{1}{L})$.  The achievable rate and RCRB exhibit different variations with the IRS beam scanning time. As the time of beam scanning increases, the IRS beamforming gain in the case of $\delta=\frac{1}{L}$ increases at the expense of reduced data transmission time, leading to an initial increase and subsequent decrease in the achievable rate. In the case of $\delta=0$, the IRS beamforming gain already reaches its maximum when $L=M$. Thus, increasing the beam scanning time further only leads to a decrease in the achievable rate. However, when considering the average aspect, the achievable rate initially rises and subsequently declines. On the other hand, the RCRB monotonically decreases  as the time of beam scanning increases. Therefore, a proper beam scanning time that balances the achievable rate and RCRB is desired.}
	
	Finally, we compare our proposed protocol with a benchmark protocol, namely the orthogonal (beam) training and sensing (OTAS) protocol, which utilizes two separate beam scanning periods for communication and sensing, respectively. Specifically, the IRS beam scanning  is only used for communication user in Phase I, followed by data transmission in Phase II. In Phase III, IRS beam scanning is conducted again for sensing only.  In the OTAS protocol, the achievable rate of the user is given by
	\begin{equation} \label{exp_ssac}
		\vspace{-0.2cm}
		R_\mathrm{OTAS} = \frac{T-\tau-\tau_s}{T} \log \left(1+ \frac{N P_t |\alpha_g|^2 |\alpha_h|^2}{\sigma^2} \frac{\sin^2 (\frac{\pi M \delta  }{2  })}{\sin^2 (\frac{\pi \delta  }{2  })} \right),
		\vspace{0.2cm}
	\end{equation}	
	where $\tau_s$ is the time of IRS beam scanning for sensing.	For ease of comparison, we set $\tau_s = \tau$.
	Fig. \ref{fig_online_w} depicts the achievable rates of STAS and OTAS protocols against the RCRB. It is observed that the achievable rate of the STAS protocol is consistently  higher than that of the OTAS protocol given the same RCRB for sensing. This is because the STAS protocol utilizes the downlink IRS beam scanning for simultaneous  communication and sensing, which saves the time for dedicated sensing compared with the OTAS protocol.

	\section{Conclusion}
	In this paper, we propose a new STAS protocol that utilizes downlink IRS beam scanning for simultaneous training and sensing to achieve efficient IRS-aided mmWave ISAC.  We derive the achievable rate of the communication user and the CRB of the target angle estimation given the IRS beam scanning codebook. The achievable rate and CRB are shown to behave differently with the time of IRS beam scanning, which, thus, needs to be properly chosen in practice to  strike a balance between communication and sensing performance. Moreover, it is shown that the STAS protocol, which utilizes IRS beam scanning for simultaneous  communication and sensing,  outperforms the benchmark OTAS protocol.
	\appendix
	\begin{appendices}
		\subsection{Proof of Theorem 1} \label{append_mle}
		Denoting $\operatorname{vec}(\mathbf{Y})$ as $\widetilde{\mathbf{y}}$, the likelihood function of $\operatorname{vec}(\mathbf{Y})$ given $\boldsymbol{\xi}$ is
		\begin{equation} 
		\hspace{-0.28cm}	L(\widetilde{\mathbf{y}}; \boldsymbol{\xi}) = \frac{1}{\left(\pi \sigma^2\right)^{L M_s}} \exp \left(-\frac{1}{\sigma^2} \|\widetilde{\mathbf{y}} - \alpha_s \operatorname{vec}(\mathbf{u}(\theta)) \|^2\right).
		\end{equation}		
		Then, maximizing the likelihood function is equivalent to minimizing $\|\widetilde{\mathbf{y}} - \alpha_s \operatorname{vec}(\mathbf{u}(\theta)) \|^2$. Therefore, the MLE can be written as
		\begin{equation} \label{exp_mle}
			\begin{aligned}
				(\theta_{\text{MLE}}, \alpha_{\text{MLE}}) = \arg \min \limits_{\theta, \alpha} \|\widetilde{\mathbf{y}} - \alpha_s \operatorname{vec}(\mathbf{u}(\theta)) \|^2.
			\end{aligned}
		\end{equation}
		With any given $\theta$, the optimal $\alpha$ is given by
		\begin{equation} \label{exp_mle_alpha}
			\alpha_{\text{MLE}} = \frac{(\operatorname{vec}(\mathbf{u}(\theta)))^H \widetilde{\mathbf{y}}}{\|\operatorname{vec}(\mathbf{u}(\theta))\|^2}.
		\end{equation}		
		Then, by substituting \eqref{exp_mle_alpha} back into  \eqref{exp_mle}, we have
		\begin{equation} \label{exp_mle_theta}
				\|\widetilde{\mathbf{y}} - \alpha_{\text{MLE}} \operatorname{vec}(\mathbf{u}(\theta)) \|^2 
				 =  \|\widetilde{\mathbf{y}}\|^2 - \frac{\left|\mathbf{a}_s^H (\theta) \mathbf{Y} \mathbf{X}^H \mathbf{q}^*(\theta) \right|^2} {L N P_t M M_s |\alpha_g|^2}.
		\end{equation}
		Thereby, the MLE of $\theta$ is given by \eqref{est_mle}.
		
		\subsection{Proof of Theorem 2} \label{append_crb}
		Since $\frac{\partial (\alpha_s\mathbf{u}(\theta) )}{\partial \theta}  = \alpha_s \operatorname{vec} (\dot{\mathbf{u}}(\theta) )$ and $\frac{\partial (\alpha_s\mathbf{u}(\theta) )}{\partial \overline{\boldsymbol{\alpha}}} = [1,j]\otimes \operatorname{vec} ({\mathbf{u}(\theta)})$, we have
		\begin{align} 
			\vspace{-0cm}
			& \mathbf{F}_{\theta\theta} =  \frac{2 |\alpha_s|^2}{\sigma^2} \operatorname{tr}  \left(\dot{\mathbf{u}}(\theta)  \dot{\mathbf{u}}^H(\theta) \right), \\
			& \mathbf{F}_{\theta \overline{\boldsymbol{\alpha}}}	
			=  \frac{2}{\sigma^2} \operatorname{Re}\left\{ \alpha_s^*  \operatorname{tr} \left(\mathbf{u}(\theta) \dot{\mathbf{u}}^H(\theta) \right)    [1,j] \right\}, \\
			&\mathbf{F}_{\overline{\boldsymbol{\alpha}} \overline{\boldsymbol{\alpha}}} 
			= \frac{2}{\sigma^2} \operatorname{tr}\left(\mathbf{u}(\theta)   \mathbf{u}^H(\theta)\right) \mathbf{I}_2. 
		\end{align}
		Thus, the FIM can be obtained as in  \eqref{exp_crb_ori}.	
	\end{appendices}
	
	\bibliographystyle{IEEEtran}
	\bibliography{review}

\begin{thebibliography}{10}
\providecommand{\url}[1]{#1}
\csname url@samestyle\endcsname
\providecommand{\newblock}{\relax}
\providecommand{\bibinfo}[2]{#2}
\providecommand{\BIBentrySTDinterwordspacing}{\spaceskip=0pt\relax}
\providecommand{\BIBentryALTinterwordstretchfactor}{4}
\providecommand{\BIBentryALTinterwordspacing}{\spaceskip=\fontdimen2\font plus
\BIBentryALTinterwordstretchfactor\fontdimen3\font minus
  \fontdimen4\font\relax}
\providecommand{\BIBforeignlanguage}[2]{{%
\expandafter\ifx\csname l@#1\endcsname\relax
\typeout{** WARNING: IEEEtran.bst: No hyphenation pattern has been}%
\typeout{** loaded for the language `#1'. Using the pattern for}%
\typeout{** the default language instead.}%
\else
\language=\csname l@#1\endcsname
\fi
#2}}
\providecommand{\BIBdecl}{\relax}
\BIBdecl

\bibitem{9737357}
F.~Liu, Y.~Cui, C.~Masouros, J.~Xu, T.~X. Han, Y.~C. Eldar, and S.~Buzzi,
  ``Integrated sensing and communications: Toward dual-functional wireless
  networks for {6G} and beyond,'' \emph{IEEE J. Sel. Areas Commun.}, vol.~40,
  no.~6, pp. 1728--1767, 2022.

\bibitem{8910627}
Q.~Wu and R.~Zhang, ``Towards smart and reconfigurable environment: Intelligent
  reflecting surface aided wireless network,'' \emph{IEEE Commun. Mag.},
  vol.~58, no.~1, pp. 106--112, 2020.

\bibitem{9326394}
Q.~Wu, S.~Zhang, B.~Zheng, C.~You, and R.~Zhang, ``Intelligent reflecting
  surface-aided wireless communications: A tutorial,'' \emph{IEEE Trans.
  Commun.}, vol.~69, no.~5, pp. 3313--3351, 2021.

\bibitem{9364358}
Z.-M. Jiang, M.~Rihan, P.~Zhang, L.~Huang, Q.~Deng, J.~Zhang, and E.~M.
  Mohamed, ``Intelligent reflecting surface aided dual-function radar and
  communication system,'' \emph{IEEE Syst. J.}, vol.~16, no.~1, pp. 475--486,
  2022.

\bibitem{9771801}
X.~Song, D.~Zhao, H.~Hua, T.~X. Han, X.~Yang, and J.~Xu, ``Joint transmit and
  reflective beamforming for {IRS}-assisted integrated sensing and
  communication,'' in \emph{Proc. IEEE WCNC}, 2022, pp. 189--194.

\bibitem{9416177}
X.~Wang, Z.~Fei, Z.~Zheng, and J.~Guo, ``Joint waveform design and passive
  beamforming for {RIS}-assisted dual-functional radar-communication system,''
  \emph{IEEE Trans. Veh. Technol.}, vol.~70, no.~5, pp. 5131--5136, 2021.

\bibitem{10038557}
X.~Cao, X.~Hu, and M.~Peng, ``Feedback-based beam training for intelligent
  reflecting surface aided {mmWave} integrated sensing and communication,''
  \emph{IEEE Trans. Veh. Technol.}, pp. 1--12, 2023.

\bibitem{9593143}
R.~S. Prasobh~Sankar, B.~Deepak, and S.~P. Chepuri, ``Joint communication and
  radar sensing with reconfigurable intelligent surfaces,'' in \emph{Proc. IEEE
  SPAWC}, 2021, pp. 471--475.

\bibitem{9724202}
X.~Shao, C.~You, W.~Ma, X.~Chen, and R.~Zhang, ``Target sensing with
  intelligent reflecting surface: Architecture and performance,'' \emph{IEEE J.
  Sel. Areas Commun.}, vol.~40, no.~7, pp. 2070--2084, 2022.

\bibitem{9367457}
L.~Liu and S.~Zhang, ``A two-stage radar sensing approach based on mimo-ofdm
  technology,'' in \emph{Proc. IEEE Globecom Workshops}, 2020, pp. 1--6.

\bibitem{kay1993fundamentals}
S.~M. Kay, \emph{Fundamentals of statistical signal processing: estimation
  theory}.\hskip 1em plus 0.5em minus 0.4em\relax Prentice-Hall, Inc., 1993.

\bibitem{9705498}
A.~Liu, Z.~Huang, M.~Li, Y.~Wan, W.~Li, T.~X. Han, C.~Liu, R.~Du, D.~K.~P. Tan,
  J.~Lu, Y.~Shen, F.~Colone, and K.~Chetty, ``A survey on fundamental limits of
  integrated sensing and communication,'' \emph{IEEE Commun. Surveys Tuts.},
  vol.~24, no.~2, pp. 994--1034, 2022.

\end{thebibliography}
	
\end{document}